\documentclass[nofootinbib,twocolumn,prl]{revtex4}

\usepackage{amsmath}
\usepackage{amsthm}
\usepackage{blindtext}
\usepackage{algorithm}
\usepackage{algorithmic}
\usepackage{amssymb}
\usepackage{graphicx}
\usepackage{color}
\usepackage{array}
\usepackage[makeroom]{cancel}
\usepackage{bbm}
\usepackage{changes}
\usepackage{changes}

\DeclareMathOperator{\sinc}{sinc}
\def\ket#1{\left| #1 \right\rangle}
\def\bra#1{\left\langle #1 \right|}

\newcommand{\beq}{\begin{equation}}
\newcommand{\eeq}{\end{equation}}

\newtheorem{lem}{Lemma}
\newtheorem{thm}{Theorem}
\newtheorem{conj}{Conjecture}

\begin{document}
\title{Quantum supremacy and high-dimensional integration}
\author{Juan Miguel Arrazola}
\affiliation{Xanadu, 372 Richmond Street W, Toronto, Ontario M5V 1X6, Canada}
\author{Patrick Rebentrost}
\affiliation{Xanadu, 372 Richmond Street W, Toronto, Ontario M5V 1X6, Canada}
\author{Christian Weedbrook}
\affiliation{Xanadu, 372 Richmond Street W, Toronto, Ontario M5V 1X6, Canada}
\begin{abstract}
We establish a connection between continuous-variable quantum computing and high-dimensional integration by showing that the outcome probabilities of continuous-variable instantaneous quantum polynomial (CV-IQP) circuits are given by integrals of oscillating functions in large dimensions. We prove two results related to the classical hardness of evaluating these integrals: (i) we show that there exist circuits such that these integrals are approximations of a weighted sum of \#P-hard problems and (ii) we prove that calculating these integrals is as hard as calculating integrals of arbitrary bounded functions. We then leverage these results to show that, given a plausible conjecture about the hardness of computing the integrals, approximate sampling from CV-IQP circuits cannot be done in polynomial time on a classical computer unless the polynomial hierarchy collapses to the third level. Our results hold even in the presence of finite squeezing and limited measurement precision, without an explicit need for fault-tolerance.
\end{abstract}
\maketitle

\textit{Introduction.---} Quantum computing is an imminent quantum technology. At the core of the efforts to build practical quantum computers is the belief that they can efficiently solve problems in cryptanalysis \cite{shor1994algorithms,proos2003shor,boneh1995quantum} and quantum simulation \cite{lloyd1996universal,lanyon2011universal,houck2012chip,
cirac2012goals,georgescu2014quantum,bernien2017probing,
zhang2017observation,berry2007efficient,berry2017exponential} for which all classical algorithms would take a forbiddingly large amount of time. Consequently, it has become vital to convincingly demonstrate that quantum computers are capable of performing tasks that are intractable for classical processors. This milestone is commonly referred to as ``quantum supremacy" \cite{harrow2017quantum}, which would result in a refutation of the Extended Church-Turing thesis. Recent efforts towards a near-term demonstration of quantum supremacy have focused on the problem of sampling from the output distribution of restricted models of quantum computing. Examples include Boson Sampling \cite{aaronson2011computational,hamilton2016gaussian}, random quantum circuits \cite{boixo2016characterizing,aaronson2016complexity}, the quantum approximate optimization algorithm \cite{farhi2014quantum,farhi2016quantum}, random Ising models \cite{gao2017quantum,bermejo2017architectures}, measurement-based quantum computing \cite{miller2017quantum} and instantaneous quantum polynomial (IQP) circuits \cite{bremner2010classical,bremner2016average,bremner2017achieving}. 

Continuous-variable (CV) quantum computing is a universal model of quantum computing where the fundamental units of information can take a continuum of possible values \cite{lloyd1999quantum,braunstein2005quantum}. This platform is ideally suited for measurement-based optical quantum computing, which provides many potential advantages compared to quantum computers manipulating qubits \cite{gu2009quantum,menicucci2006universal}. Progress in characterizing quantum supremacy for CV quantum computers has recently been addressed, notably in Ref. \cite{douce2017continuous}, where it was shown that any classical algorithm that can exactly sample from any fault-tolerant CV-IQP circuit must take exponential time unless the polynomial hierarchy collapses to third level. Nevertheless, several important questions remain open. For instance, it is crucial to determine whether the hardness result remains even for approximate simulation of the circuits and whether fault-tolerance is needed in CV-IQP circuits to demonstrate quantum supremacy. It is also of great interest to understand if CV-IQP circuits can be related to problems of practical significance.

In this work, we connect the hardness of sampling from CV-IQP circuits to the difficulty of computing integrals of oscillating functions in a large number of dimensions. High-dimensional integration is an important and widely-studied problem in many areas of physics, chemistry, finance, and statistics. Although several techniques are known for efficiently calculating one-dimensional integrals, extending them to many variables suffers from the so-called ``curse of dimensionality". This is what makes one-dimensional strategies ineffective for the high-dimensional case, where general integrals require exponential resources to be evaluated \cite{stroud1971approximate,sloan1998quasi,novak2009approximation,novak2008tractability1,novak2010tractability2, curse}. In fact, it has already been shown that certain integrals arising in the description of Boson Sampling circuits are \#P-hard to calculate \cite{rohde2016quantum}. Thus, although they are not the preferred problem of computer scientists, integrals of functions over many variables have been extensively studied with no known efficient algorithms known for arbitrary integrals.

We prove that there exist CV-IQP circuits for which the corresponding integrals are an arbitrarily good approximation of a weighted sum of independent \#P-hard problems. Furthermore, we show that evaluating these integrals is as hard as for arbitrary bounded functions, which are known to require exponential time to approximate on a classical computer in a worst-case setting\cite{curse,novak2008tractability1,novak2010tractability2}. We then prove that if these integrals are \#P-hard to approximate on average, then any classical algorithm for approximate sampling from the output of CV-IQP circuits must run in exponential time unless the polynomial hierarchy collapses to third level. We conclude by showing that our results hold even if the input states are finitely squeezed and the measurements have finite precision, without an explicit need for fault-tolerance.

\textit{CV-IQP circuits.---} Continuous-variable instantaneous quantum polynomial (CV-IQP) circuits are a subclass of circuits on a continuous-variable quantum computer which can be decomposed as follows \cite{douce2017continuous}: (i) Inputs states are momentum-squeezed vacuum states, (ii) Unitary transformations are diagonal in the position quadrature (iii) measurements are homodyne momentum measurements. A generic CV-IQP circuit is illustrated in Fig. \ref{Fig: Circuit}.

We denote the position eigenstates of $n$ optical qumodes as $\ket{q}=\ket{q_1}\ket{q_2}\ldots\ket{q_n}$ with $q\in\mathbb{R}^n$ and consider circuits with diagonal gates $C_f$ acting on position eigenstates as $C_f\ket{q}=e^{i f(q)}\ket{q}$, where $f:\mathbb{R}^n\rightarrow\mathbb{R}$ is a polynomial. In the ideal case, the probability amplitude of obtaining an outcome $s=(s_1,s_2,\ldots,s_n)$ is given by
\begin{align}\label{Amplitude}
A_f(s)&=\frac{1}{(2\pi)^n}\int_{\mathbb{R}^n}e^{i f(q)}e^{-i s\cdot q}dq^n,
\end{align}
where $s\cdot q=\sum_k s_kq_k$ is the inner product. The probability of outcome $s$ is $P(s)=|A_f(s)|^2$. We refer to this expression as a \emph{CV-IQP integral.} Note that $A_f(s)$ is the Fourier transform of $e^{if(q)}$ and therefore the CV-IQP circuit is sampling from a distribution induced by this Fourier transform. Based on the vast literature on high-dimensional integration, it is reasonable to expect CV-IQP integrals to be intractable to approximate for general circuits where $f(q)$ is a high-degree polynomial. Note that for polynomials of degree 2, the circuits can be efficiently simulated classically \cite{bartlett2002efficient}. In the following, we formalize this intuition by proving two results regarding the computational complexity of CV-IQP integrals. 

\begin{center}
\begin{figure}[t!]
\includegraphics[width= \columnwidth]{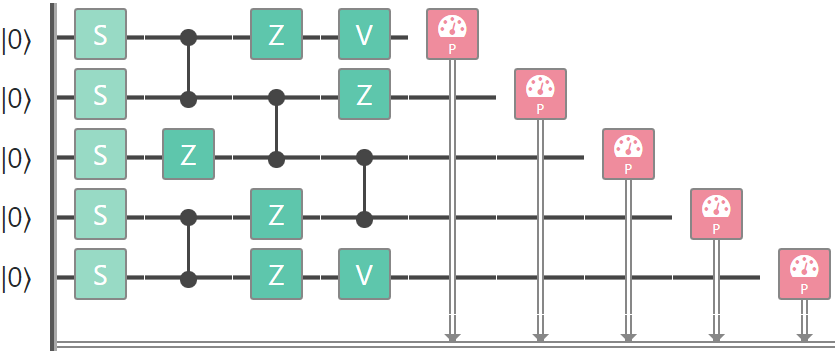}
\caption{Schematic representation of a CV-IQP circuit acting on five qumodes. Vacuum states are squeezed in the momentum quadrature by the action of squeezers $S$. A diagonal unitary transformation $C_f$ is applied, which in this case is decomposed in terms of Z gates, controlled-Z gates and cubic phase gates $V$. Finally, a momentum homodyne measurement is performed on each individual qumode. }\label{Fig: Circuit}
\end{figure}
\end{center}

\textit{CV-IQP integrals as weighted sums of \#P-hard problems.---} We begin by making a series of appropriate approximations to the integral of Eq. \eqref{Amplitude}.  The integrand is highly oscillatory for large values of $q$ leading to contributions that average to zero. Therefore, we can restrict the integration to the hypercube $D_L=[-L,L]^n$ for some appropriate $L$. Furthermore, we can approximate the integrand by a sum of indicator functions and replace the integral by a Riemann sum. As shown in the Supplemental Material, this leads to the expression 
\begin{align}\label{P_{f,N}}
A_f(s)&=\frac{\Delta q^n}{(2\pi)^n} \sum_{b=0}^{N-1}\sum_{q} e^{i \phi_b}\Theta_b^{f}(q)+\epsilon,
\end{align}
where $q$ takes on a finite amount of values with a spacing $\Delta q$, $\epsilon$ is an arbitrarily small approximation error, and we have defined the angles $\phi_b=\frac{ 2\pi b }{N}$ as well as the indicator functions
\beq
\Theta_b^{f}(q)=\begin{cases}
1& \textrm{ if } f(q)-s\cdot q\in [\phi_b,\phi_{b+1})\mod 2\pi\\
0 & \textrm{ otherwise.}
\end{cases}
\eeq
Since there are only finitely many of them, we can associate each vector $q$ with an $m$-bit string $x$ and view each indicator function as a Boolean function $\Theta_b^{f}(x)$. In this case, we can write 
\begin{align}
A_f(s)&=\frac{\Delta q^n}{(2\pi)^n} \sum_{b=0}^{N-1}e^{i\phi_b}\Omega^f_b+\epsilon,\label{A_f(s)}\\
\Omega^{f}_b:&=\sum_{x\in\{0,1\}^m}\Theta_b^{f}(x).
\end{align}
By definition, calculating $\Omega_b$ is in the complexity class \#P. In fact, the problem of computing the approximation in Eq. \eqref{A_f(s)} belongs to a class that is the generalization of the class GapP (the closure of \#P under substraction) where we allow $N$ different complex phases $e^{i\phi_b}$ instead of only $e^{i0}=1$ and $e^{i\pi}=-1$. Such a class may be of independent interest in complexity theory.

Since we have made no assumptions about $f(q)$, we are free to choose them such that the corresponding quantities $\Omega^{f}_b$ are as hard to compute as possible. Let $\Phi_0,\Phi_1,\ldots,\Phi_{N-1}$ be Boolean functions $\Phi_b:\{0,1\}^{m-l}\rightarrow \{0,1\}$ such that computing $\sum_{y}\Phi_b(y)$ is a \#P-hard problem for all $b$. Denote by $x=(b,y)$ an $m$-bit string where the first $l$ bits are a binary representation of $b$ and the remaining $m-l$ are an arbitrary string $y$. Define the function 
\beq
f(x)=f(b,y)= \sum_{b'=0}^{N-1} \frac{2\pi(b'+\frac12)}{N}\Phi_{b'}(y)\delta_{b,b'},
\eeq
where $\delta_{b,b'}$ is the Kronecker delta function. It then holds that $\Theta_b^{f}(b',y)=\Phi_b(y)\delta_{b,b'}$ and in particular we can write
\beq
A_f(s)=\frac{\Delta q^n}{(2\pi)^n}\sum_{b=0}^{N-1}\sum_{y\in\{0,1\}^{m-l}}e^{i \phi_b}\Phi_b(y),
\eeq
i.e. $A_f(s)$ is a weighted sum of quantities that are \#P-hard to compute. This is a strong indication that CV-IQP integrals can be \#P-hard to compute.

\textit{CV-IQP integrals as integrals of bounded functions.---} Consider the real part of the integral in Eq. \eqref{Amplitude}, given by 
\beq
\text{Re}[A_f(s)]=\frac{1}{(2\pi)^n}\int_{\mathbb{R}^n}\cos[f_s(q)]dq^n,
\eeq 
where $f_s(q)=f(q)-s\cdot q$. For any real, bounded function $\phi(q)$ satisfying $|\phi(q)|\leq c$ for some $c>0$, we can set $f_s(q)=\cos^{-1}[\phi(q)/c]$ such that the real part of the CV-IQP integral is proportional to the integral over $\phi(q)$
\beq
\text{Re}[A_f(s)]=\frac{1}{(2\pi)^n}\int_{\mathbb{R}^n}\frac{1}{c}\phi(q)dq^n.
\eeq
Since calculating $A_f(s)$ is at least as hard as computing its real part, we conclude that CV-IQP integrals are as hard to compute as integrals of any real, bounded function. For example, in Ref. \cite{curse}, it was shown that numerical integration using deterministic algorithms requires a worst-case number of function evaluations $N$ that is exponential in the dimension of the integral, i.e., $N=2^{O(n)}$. The proof of this fact, which we reproduce in the Supplemental Material, is based on the integration of fooling functions $\phi(q)$ satisfying $|\phi(q)|\leq 1$. From the above discussion, we can design CV-IQP integrals that are as hard to integrate as these fooling functions. Therefore, we conclude that deterministic numerical algorithms to evaluate CV-IQP integrals require exponential time on a classical computer in a worst-case setting. Together with our previous results, this cements the understanding that calculating CV-IQP integrals is a challenging computational problem. 

\textit{Hardness of sampling from CV-IQP circuits.---}
We have previously presented arguments to establish the computational hardness of approximating CV-IQP integrals. We now formulate this concretely in the form of the following conjecture. 
\begin{conj}\label{conjecture}
There exists a family of polynomials $\mathcal{F}$ and a corresponding family of CV-IQP circuits $\mathcal{C}$ such that, for a fraction $1/8$ of circuits $C_f\in \mathcal{C}$ with $f\in \mathcal{F}$, it is a \#P-hard problem to approximate the probability $|A_f(0)|^2$ of obtaining outcome $s=(0,\ldots,0)$ up to multiplicative error $1/4+o(1)$ .
\end{conj}

Our goal is to leverage this conjecture into a statement about the computational hardness of approximate sampling from the output distribution of CV-IQP circuits. We will make use of the following Lemma, adapted from Ref. \cite{bremner2016average} into a CV setting by using displacement circuits to permute the outcome probabilities.

\begin{lem}\label{lemma_main} Let $C_f$ be a CV-IQP circuit acting on $n$ qumodes, where $C_f$ is chosen from some appropriate family of circuits. Let $C_{f,r}$ be the circuit obtained by adding diagonal gates $U_r=\prod_{k=1}^ne^{-i \hat{q}_k r_k}$ to $C_f$, with $r=(r_1,r_2,\ldots,r_n)$ and $r_k\in\{-L,-L+2\Delta_p,\ldots,L-2\Delta_p,L\}$ for some $L> 1$. There are $\ell=L/\Delta_p$ possible values of each $r_k$. Assume that there exists a polynomial-time classical algorithm $\mathcal{A}$ that can approximate the probability distribution of any CV-IQP circuit $C'_f$ up to additive error $\epsilon$. Then for any $\delta>0$, there exists an $\text{FBPP}^{\text{NP}}$ algorithm that given access to $\mathcal{A}$ approximates $|A_{f,r}(0)|^2$ for a circuit $C_{f,r}$ up to additive error
\beq
O\left(\frac{(1+o(1))\epsilon}{\delta \ell^n}+|A_{f,r}(0)|^2/\text{poly(n)}\right)
\eeq
with probability at least $1-\delta$.\\
\end{lem}
See the Supplemental Material for a proof. Our goal is to show that this Lemma also implies that the $\text{FBPP}^{\text{NP}}$ algorithm gives a good \emph{multiplicative} approximation of CV-IQP integrals. For this we require an anti-concentration result, namely we need to show that $|A_f(0)|^2\geq \beta(\ell^{-n})$ for some $\beta>0$. From the Payley-Zigmund inequality, it holds that 
\beq\label{Payley}
\Pr(|A_f(0)|^2\geq \alpha\mathbb{E}[|A_f(0)|^2])\geq (1-\alpha)^2\frac{\mathbb{E}[|A_f(0)|^2]^2}{\mathbb{E}[|A_f(0)|^4]},
\eeq
where the expectation is taken over all circuits $C_{f,s}$. As shown in the Supplemental Material, there is a value of $L$ such that $\mathbb{E}[|A_f(0)|^2]^2/\mathbb{E}[|A_f(0)|^4]\geq 1$ and therefore for $\alpha=1/2$ it holds that
\beq
\Pr(|A_f(0)|^2\geq \frac12\ell^{-n})\geq \frac{1}{4}.
\eeq
For $\epsilon=\alpha(1-\alpha)^2/8$ and $\delta=(1-\alpha)^2/2$, Lemma \ref{lemma_main} implies that there exists an $\text{FBPP}^{\text{NP}}$ algorithm with access to the classical sampling algorithm $\mathcal{A}$ that, with constant probability over the choice of circuits, approximates $|A_f(0)|^2$ up to additive error 
\begin{align}\label{multip. error}
O\left((1+o(1))\frac{1}{4}+1/\text{poly(n)}\right)|A_f(0)|^2,
\end{align}
and therefore it also approximates $|A_f(0)|^2$ up to constant multiplicative error $1/4+o(1)$. See the Supplemental Material for full details of this derivation. We are now ready to state the main result of this section.

\begin{thm}\label{theorem}
Assume that Conjecture \ref{conjecture} is true. Then, if there exists a classical algorithm running in polynomial time that samples from any CV-IQP circuit up to additive error $1/64$, the polynomial hierarchy collapses to third level.
\end{thm}
\begin{proof}
From Lemma \ref{lemma_main} and Eq. \eqref{multip. error}, a polynomial-time classical algorithm that samples from any CV-IQP circuit up to additive error $1/64$ implies an $\text{FBBP}^\text{NP}$ algorithm for approximating $|A_f(0)|^2$ up to a multiplicative error of $1/4+o(1)$ for at least a fraction $1/8$ of circuits. From Conjecture \ref{conjecture}, this implies that the $\text{FBBP}^\text{NP}$ algorithm, which is contained in the third level of the polynomial hierarchy, can solve any problem in $\text{P}^{\#P}$. By Toda's theorem \cite{toda1991pp}, the entire polynomial hierarchy is contained in $\text{P}^{\#\text{P}}$ and therefore this causes a collapse to third level.
\end{proof}
\begin{center}
\begin{figure*}[t!]
\begin{tabular}{ccccc}
\includegraphics[width=0.4 \columnwidth]{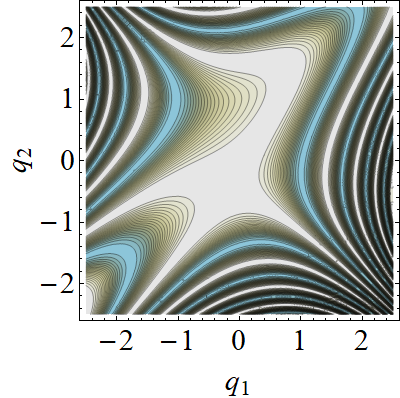}& \includegraphics[width=0.4 \columnwidth]{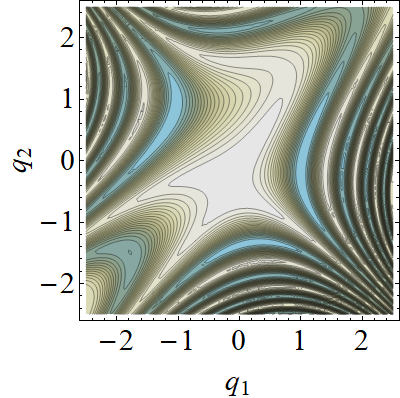} &
\includegraphics[width=0.4 \columnwidth]{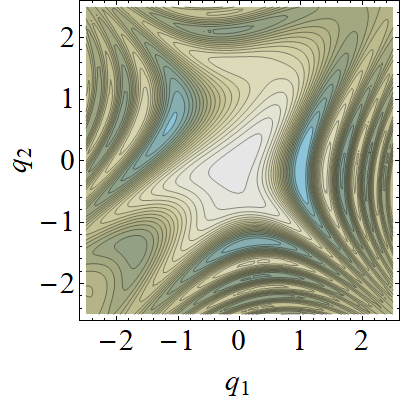}&
\includegraphics[width=0.4 \columnwidth]{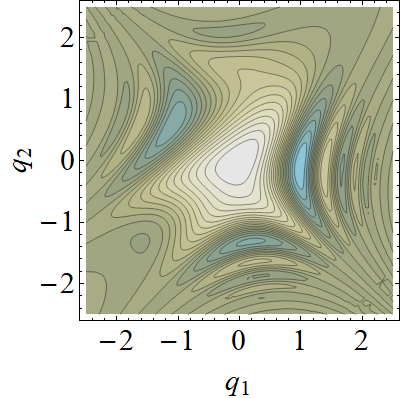} &
\includegraphics[width=0.1058 \columnwidth]{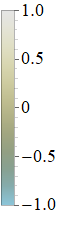}
\end{tabular}
\caption{Contour plot of the real part of the integrand for the polynomial $f(q_1,q_2)=q_1-q_2+q_1q_2+q_1^2-q_2^2-q_1q_2^2-q_1^2q_2+q_1^3+q_2^3$. The left panel shows the ideal case of infinte squeezing and precision, whereas the remainder panels show, from left to right, the case for precision $\Delta_p=10^{-3}$ and squeezing $\sigma=3$, $\sigma=1.5$ and $\sigma=1$ . In terms of decibels (dBs), using $S_{dB}=10\log_{10}\sigma^2$, these approximately correspond to 9.5 dB, 3.5 dB, and 0 dB respectively. For large values of $\sigma$, the function is approximately equal to the ideal case in the region of slow oscillations, but it becomes closer to functions that are easier to integrate when squeezing is low.}\label{Fig: Int1}
\end{figure*}
\end{center}
\textit{Role of finite squeezing.---} When the inputs are finitely squeezed states with variance $\sigma^2$, and the measurements have limited precision $\Delta_p$ as given by the projectors $\eta_{s_k}=\int_{s_k-\Delta_p}^{s_k+\Delta_p}\ket{p_k}\bra{p_k}dp_k$, in the regime of $\sigma\ll 1/\Delta_p$, the probability $P_f(s)$ of obtaining an outcome $s$ can be expressed as $\tilde{P}_f(s)=|\tilde{A}_f(s)|^2$ where
\beq\label{Amplitude_finite}
\tilde{A}_f(s)=\left(\frac{\Delta_p}{\pi^{3/2}\sigma}\right)^{n/2}\int_{\mathbb{R}^n}e^{i f(q)-i s\cdot q}e^{-q^2/(2\sigma^2)}dq^n.
\eeq
Besides normalization factors, the only difference compared to the integrals in the ideal case is the presence of a Gaussian term $e^{-q^2/(2\sigma^2)}$, which sets a scale $\sigma$ for the region where the integrand is non-negligible. As discussed previously, the fast oscillations of the function $e^{if(q)}$ also introduce a region with non-negligible contributions to the integral. This region is defined by a scale $L_{\text{osc}}$. The integrals in the ideal and the finite squeezing cases are then excellent approximations of each other as long as $\sigma$ is sufficiently larger than $L_{\text{osc}}$, retaining their computational hardness. For small $\sigma$, the integrals are close to Gaussian integrals, which can be computed efficiently. This is illustrated in Fig. \ref{Fig: Int1} for the case where the function $f(q)$ is a degree-3 polynomials over two variables.

Finally, we note that there is a trade-off between the amount of squeezing in the initial states and the coefficients of the polynomial $f(q)$. A transformation $q\rightarrow T q$ for some $T>1$ induces a change $L_{\text{osc}}\rightarrow L_{\text{osc}}/T$, allowing us to make the region of slow oscillations as small as desired compared to $\sigma$. This is illustrated in Fig. \ref{Fig: Int3}. 


\begin{center}
\begin{figure}[t!]
\begin{tabular}{cc}
\includegraphics[width=0.5 \columnwidth]{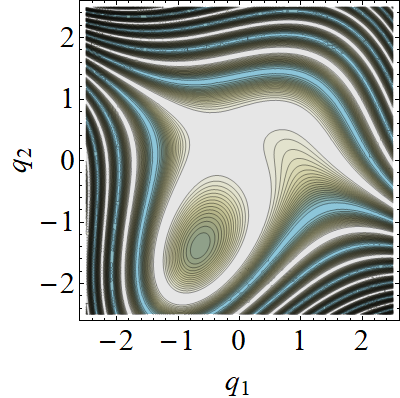} & \includegraphics[width=0.5 \columnwidth]{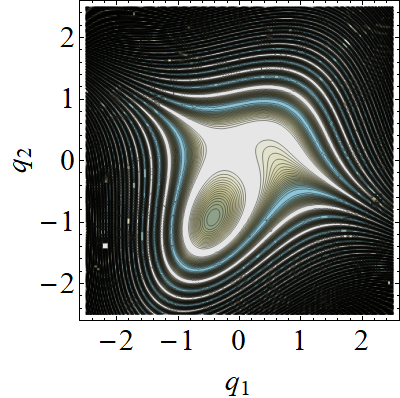}
\end{tabular}
\caption{Contour plot of the real part of the integrand for the polynomial $f(q_1,q_2)=-Tq_1-Tq_2+T^2q_1q_2-T^2q_1^2+T^2q_2^2+T^3q_1q_2^2-T^3q_1^2q_2+T^3q_1^3+T^3q_2^3$. The color code for the contours is the same as in Fig. \ref{Fig: Int1}. The left panel shows the case for $T=1$ and the right panel shows the case for $T=1.5$.  The effect of a larger value of $T$ is to rescale the function so that the region that contributes to the integral is reduced to a smaller area, which can potentially compensate for lesser amounts of squeezing and vice versa. }\label{Fig: Int3}
\end{figure}
\end{center}

\textit{Discussion.---} We have shown that, provided CV-IQP integrals are on average \#P-hard to approximate, approximate sampling from the output distribution of CV-IQP circuits cannot be done in polynomial time on a classical computer unless the polynomial hierarchy collapses to third level. The conjecture that these integrals are \#P-hard to approximate is not only supported by our results connecting CV-IQP integrals to computational complexity theory, but also by decades of research on attempts to efficiently compute high-dimensional integrals. Our results thus further supports the claim that continuous-variable quantum computers are candidates for challenging the Extended Church-Turing thesis by demonstrating quantum supremacy in the near term. Crucially, this holds even in the case of approximate simulation of CV-IQP circuits with finitely-squeezed input states and limited precision, without the explicit need for fault-tolerance and error correction. 

To strengthen the claim of supremacy even further, it is important to extend our results to the case where there are errors in the diagonal circuit $C_f\ket{q}=e^{if(q)}\ket{q}$ and to show that CV-IQP integrals remain hard to calculate for a simple class of circuits, for example those corresponding to degree-3 polynomials. Finally, it is of great interest to understand the extent to which CV-IQP circuits may be able to directly solve challenging computational problems. Indeed, as we have previously discussed, high dimensional integrals appear in a large class of problems of practical interest -- notably in physics and finance -- making these a potentially fertile ground for applications of continuous-variable quantum computing.

\textit{Acknowledgements.---} The authors thank A. Ignjatovic, T. Bromley, and N. Killoran for valuable discussions.

\bibliographystyle{apsrev}
\bibliography{Bibliography}

\begin{thebibliography}{40}
\expandafter\ifx\csname natexlab\endcsname\relax\def\natexlab#1{#1}\fi
\expandafter\ifx\csname bibnamefont\endcsname\relax
  \def\bibnamefont#1{#1}\fi
\expandafter\ifx\csname bibfnamefont\endcsname\relax
  \def\bibfnamefont#1{#1}\fi
\expandafter\ifx\csname citenamefont\endcsname\relax
  \def\citenamefont#1{#1}\fi
\expandafter\ifx\csname url\endcsname\relax
  \def\url#1{\texttt{#1}}\fi
\expandafter\ifx\csname urlprefix\endcsname\relax\def\urlprefix{URL }\fi
\providecommand{\bibinfo}[2]{#2}
\providecommand{\eprint}[2][]{\url{#2}}

\bibitem[{\citenamefont{Shor}(1994)}]{shor1994algorithms}
\bibinfo{author}{\bibfnamefont{P.~W.} \bibnamefont{Shor}}, in
  \emph{\bibinfo{booktitle}{Foundations of Computer Science, 1994 Proceedings.,
  35th Annual Symposium on}} (\bibinfo{organization}{Ieee},
  \bibinfo{year}{1994}), pp. \bibinfo{pages}{124--134}.

\bibitem[{\citenamefont{Proos and Zalka}(2003)}]{proos2003shor}
\bibinfo{author}{\bibfnamefont{J.}~\bibnamefont{Proos}} \bibnamefont{and}
  \bibinfo{author}{\bibfnamefont{C.}~\bibnamefont{Zalka}},
  \bibinfo{journal}{Quantum Information \& Computation}
  \textbf{\bibinfo{volume}{3}}, \bibinfo{pages}{317} (\bibinfo{year}{2003}).

\bibitem[{\citenamefont{Boneh and Lipton}(1995)}]{boneh1995quantum}
\bibinfo{author}{\bibfnamefont{D.}~\bibnamefont{Boneh}} \bibnamefont{and}
  \bibinfo{author}{\bibfnamefont{R.~J.} \bibnamefont{Lipton}}, in
  \emph{\bibinfo{booktitle}{Annual International Cryptology Conference}}
  (\bibinfo{organization}{Springer}, \bibinfo{year}{1995}), pp.
  \bibinfo{pages}{424--437}.

\bibitem[{\citenamefont{Lloyd et~al.}(1996)}]{lloyd1996universal}
\bibinfo{author}{\bibfnamefont{S.}~\bibnamefont{Lloyd}} \bibnamefont{et~al.},
  \bibinfo{journal}{Science} \textbf{\bibinfo{volume}{273}},
  \bibinfo{pages}{1073} (\bibinfo{year}{1996}).

\bibitem[{\citenamefont{Lanyon et~al.}(2011)\citenamefont{Lanyon, Hempel, Nigg,
  M{\"u}ller, Gerritsma, Z{\"a}hringer, Schindler, Barreiro, Rambach, Kirchmair
  et~al.}}]{lanyon2011universal}
\bibinfo{author}{\bibfnamefont{B.~P.} \bibnamefont{Lanyon}},
  \bibinfo{author}{\bibfnamefont{C.}~\bibnamefont{Hempel}},
  \bibinfo{author}{\bibfnamefont{D.}~\bibnamefont{Nigg}},
  \bibinfo{author}{\bibfnamefont{M.}~\bibnamefont{M{\"u}ller}},
  \bibinfo{author}{\bibfnamefont{R.}~\bibnamefont{Gerritsma}},
  \bibinfo{author}{\bibfnamefont{F.}~\bibnamefont{Z{\"a}hringer}},
  \bibinfo{author}{\bibfnamefont{P.}~\bibnamefont{Schindler}},
  \bibinfo{author}{\bibfnamefont{J.}~\bibnamefont{Barreiro}},
  \bibinfo{author}{\bibfnamefont{M.}~\bibnamefont{Rambach}},
  \bibinfo{author}{\bibfnamefont{G.}~\bibnamefont{Kirchmair}},
  \bibnamefont{et~al.}, \bibinfo{journal}{Science}
  \textbf{\bibinfo{volume}{334}}, \bibinfo{pages}{57} (\bibinfo{year}{2011}).

\bibitem[{\citenamefont{Houck et~al.}(2012)\citenamefont{Houck, T{\"u}reci, and
  Koch}}]{houck2012chip}
\bibinfo{author}{\bibfnamefont{A.~A.} \bibnamefont{Houck}},
  \bibinfo{author}{\bibfnamefont{H.~E.} \bibnamefont{T{\"u}reci}},
  \bibnamefont{and} \bibinfo{author}{\bibfnamefont{J.}~\bibnamefont{Koch}},
  \bibinfo{journal}{Nature Physics} \textbf{\bibinfo{volume}{8}},
  \bibinfo{pages}{292} (\bibinfo{year}{2012}).

\bibitem[{\citenamefont{Cirac and Zoller}(2012)}]{cirac2012goals}
\bibinfo{author}{\bibfnamefont{J.~I.} \bibnamefont{Cirac}} \bibnamefont{and}
  \bibinfo{author}{\bibfnamefont{P.}~\bibnamefont{Zoller}},
  \bibinfo{journal}{Nature Physics} \textbf{\bibinfo{volume}{8}},
  \bibinfo{pages}{264} (\bibinfo{year}{2012}).

\bibitem[{\citenamefont{Georgescu et~al.}(2014)\citenamefont{Georgescu, Ashhab,
  and Nori}}]{georgescu2014quantum}
\bibinfo{author}{\bibfnamefont{I.}~\bibnamefont{Georgescu}},
  \bibinfo{author}{\bibfnamefont{S.}~\bibnamefont{Ashhab}}, \bibnamefont{and}
  \bibinfo{author}{\bibfnamefont{F.}~\bibnamefont{Nori}},
  \bibinfo{journal}{Reviews of Modern Physics} \textbf{\bibinfo{volume}{86}},
  \bibinfo{pages}{153} (\bibinfo{year}{2014}).

\bibitem[{\citenamefont{Bernien et~al.}(2017)\citenamefont{Bernien, Schwartz,
  Keesling, Levine, Omran, Pichler, Choi, Zibrov, Endres, Greiner
  et~al.}}]{bernien2017probing}
\bibinfo{author}{\bibfnamefont{H.}~\bibnamefont{Bernien}},
  \bibinfo{author}{\bibfnamefont{S.}~\bibnamefont{Schwartz}},
  \bibinfo{author}{\bibfnamefont{A.}~\bibnamefont{Keesling}},
  \bibinfo{author}{\bibfnamefont{H.}~\bibnamefont{Levine}},
  \bibinfo{author}{\bibfnamefont{A.}~\bibnamefont{Omran}},
  \bibinfo{author}{\bibfnamefont{H.}~\bibnamefont{Pichler}},
  \bibinfo{author}{\bibfnamefont{S.}~\bibnamefont{Choi}},
  \bibinfo{author}{\bibfnamefont{A.~S.} \bibnamefont{Zibrov}},
  \bibinfo{author}{\bibfnamefont{M.}~\bibnamefont{Endres}},
  \bibinfo{author}{\bibfnamefont{M.}~\bibnamefont{Greiner}},
  \bibnamefont{et~al.}, \bibinfo{journal}{Nature}
  \textbf{\bibinfo{volume}{551}}, \bibinfo{pages}{579} (\bibinfo{year}{2017}).

\bibitem[{\citenamefont{Zhang et~al.}(2017)\citenamefont{Zhang, Pagano, Hess,
  Kyprianidis, Becker, Kaplan, Gorshkov, Gong, and
  Monroe}}]{zhang2017observation}
\bibinfo{author}{\bibfnamefont{J.}~\bibnamefont{Zhang}},
  \bibinfo{author}{\bibfnamefont{G.}~\bibnamefont{Pagano}},
  \bibinfo{author}{\bibfnamefont{P.~W.} \bibnamefont{Hess}},
  \bibinfo{author}{\bibfnamefont{A.}~\bibnamefont{Kyprianidis}},
  \bibinfo{author}{\bibfnamefont{P.}~\bibnamefont{Becker}},
  \bibinfo{author}{\bibfnamefont{H.}~\bibnamefont{Kaplan}},
  \bibinfo{author}{\bibfnamefont{A.~V.} \bibnamefont{Gorshkov}},
  \bibinfo{author}{\bibfnamefont{Z.-X.} \bibnamefont{Gong}}, \bibnamefont{and}
  \bibinfo{author}{\bibfnamefont{C.}~\bibnamefont{Monroe}},
  \bibinfo{journal}{Nature} \textbf{\bibinfo{volume}{551}},
  \bibinfo{pages}{601} (\bibinfo{year}{2017}).

\bibitem[{\citenamefont{Berry et~al.}(2007)\citenamefont{Berry, Ahokas, Cleve,
  and Sanders}}]{berry2007efficient}
\bibinfo{author}{\bibfnamefont{D.~W.} \bibnamefont{Berry}},
  \bibinfo{author}{\bibfnamefont{G.}~\bibnamefont{Ahokas}},
  \bibinfo{author}{\bibfnamefont{R.}~\bibnamefont{Cleve}}, \bibnamefont{and}
  \bibinfo{author}{\bibfnamefont{B.~C.} \bibnamefont{Sanders}},
  \bibinfo{journal}{Communications in Mathematical Physics}
  \textbf{\bibinfo{volume}{270}}, \bibinfo{pages}{359} (\bibinfo{year}{2007}).

\bibitem[{\citenamefont{Berry et~al.}(2017)\citenamefont{Berry, Childs, Cleve,
  Kothari, and Somma}}]{berry2017exponential}
\bibinfo{author}{\bibfnamefont{D.~W.} \bibnamefont{Berry}},
  \bibinfo{author}{\bibfnamefont{A.~M.} \bibnamefont{Childs}},
  \bibinfo{author}{\bibfnamefont{R.}~\bibnamefont{Cleve}},
  \bibinfo{author}{\bibfnamefont{R.}~\bibnamefont{Kothari}}, \bibnamefont{and}
  \bibinfo{author}{\bibfnamefont{R.~D.} \bibnamefont{Somma}}, in
  \emph{\bibinfo{booktitle}{Forum of Mathematics, Sigma}}
  (\bibinfo{organization}{Cambridge University Press}, \bibinfo{year}{2017}),
  vol.~\bibinfo{volume}{5}.

\bibitem[{\citenamefont{Harrow and Montanaro}(2017)}]{harrow2017quantum}
\bibinfo{author}{\bibfnamefont{A.~W.} \bibnamefont{Harrow}} \bibnamefont{and}
  \bibinfo{author}{\bibfnamefont{A.}~\bibnamefont{Montanaro}},
  \bibinfo{journal}{Nature} \textbf{\bibinfo{volume}{549}},
  \bibinfo{pages}{203} (\bibinfo{year}{2017}).

\bibitem[{\citenamefont{Aaronson and
  Arkhipov}(2011)}]{aaronson2011computational}
\bibinfo{author}{\bibfnamefont{S.}~\bibnamefont{Aaronson}} \bibnamefont{and}
  \bibinfo{author}{\bibfnamefont{A.}~\bibnamefont{Arkhipov}}, in
  \emph{\bibinfo{booktitle}{Proceedings of the forty-third annual ACM symposium
  on Theory of computing}} (\bibinfo{organization}{ACM}, \bibinfo{year}{2011}),
  pp. \bibinfo{pages}{333--342}.

\bibitem[{\citenamefont{Hamilton et~al.}(2017)\citenamefont{Hamilton, Kruse,
  Sansoni, Barkhofen, Silberhorn, and Jex}}]{hamilton2016gaussian}
\bibinfo{author}{\bibfnamefont{C.~S.} \bibnamefont{Hamilton}},
  \bibinfo{author}{\bibfnamefont{R.}~\bibnamefont{Kruse}},
  \bibinfo{author}{\bibfnamefont{L.}~\bibnamefont{Sansoni}},
  \bibinfo{author}{\bibfnamefont{S.}~\bibnamefont{Barkhofen}},
  \bibinfo{author}{\bibfnamefont{C.}~\bibnamefont{Silberhorn}},
  \bibnamefont{and} \bibinfo{author}{\bibfnamefont{I.}~\bibnamefont{Jex}},
  \bibinfo{journal}{Physical Review Letters} \textbf{\bibinfo{volume}{119}},
  \bibinfo{pages}{170501} (\bibinfo{year}{2017}).

\bibitem[{\citenamefont{Boixo et~al.}(2016)\citenamefont{Boixo, Isakov,
  Smelyanskiy, Babbush, Ding, Jiang, Martinis, and
  Neven}}]{boixo2016characterizing}
\bibinfo{author}{\bibfnamefont{S.}~\bibnamefont{Boixo}},
  \bibinfo{author}{\bibfnamefont{S.~V.} \bibnamefont{Isakov}},
  \bibinfo{author}{\bibfnamefont{V.~N.} \bibnamefont{Smelyanskiy}},
  \bibinfo{author}{\bibfnamefont{R.}~\bibnamefont{Babbush}},
  \bibinfo{author}{\bibfnamefont{N.}~\bibnamefont{Ding}},
  \bibinfo{author}{\bibfnamefont{Z.}~\bibnamefont{Jiang}},
  \bibinfo{author}{\bibfnamefont{J.~M.} \bibnamefont{Martinis}},
  \bibnamefont{and} \bibinfo{author}{\bibfnamefont{H.}~\bibnamefont{Neven}},
  \bibinfo{journal}{arXiv:1608.00263}  (\bibinfo{year}{2016}).

\bibitem[{\citenamefont{Aaronson and Chen}(2016)}]{aaronson2016complexity}
\bibinfo{author}{\bibfnamefont{S.}~\bibnamefont{Aaronson}} \bibnamefont{and}
  \bibinfo{author}{\bibfnamefont{L.}~\bibnamefont{Chen}},
  \bibinfo{journal}{arXiv:1612.05903}  (\bibinfo{year}{2016}).

\bibitem[{\citenamefont{Farhi et~al.}(2014)\citenamefont{Farhi, Goldstone, and
  Gutmann}}]{farhi2014quantum}
\bibinfo{author}{\bibfnamefont{E.}~\bibnamefont{Farhi}},
  \bibinfo{author}{\bibfnamefont{J.}~\bibnamefont{Goldstone}},
  \bibnamefont{and} \bibinfo{author}{\bibfnamefont{S.}~\bibnamefont{Gutmann}},
  \bibinfo{journal}{arXiv:1411.4028}  (\bibinfo{year}{2014}).

\bibitem[{\citenamefont{Farhi and Harrow}(2016)}]{farhi2016quantum}
\bibinfo{author}{\bibfnamefont{E.}~\bibnamefont{Farhi}} \bibnamefont{and}
  \bibinfo{author}{\bibfnamefont{A.~W.} \bibnamefont{Harrow}},
  \bibinfo{journal}{arXiv:1602.07674}  (\bibinfo{year}{2016}).

\bibitem[{\citenamefont{Gao et~al.}(2017)\citenamefont{Gao, Wang, and
  Duan}}]{gao2017quantum}
\bibinfo{author}{\bibfnamefont{X.}~\bibnamefont{Gao}},
  \bibinfo{author}{\bibfnamefont{S.-T.} \bibnamefont{Wang}}, \bibnamefont{and}
  \bibinfo{author}{\bibfnamefont{L.-M.} \bibnamefont{Duan}},
  \bibinfo{journal}{Physical Review Letters} \textbf{\bibinfo{volume}{118}},
  \bibinfo{pages}{040502} (\bibinfo{year}{2017}).

\bibitem[{\citenamefont{Bermejo-Vega et~al.}(2017)\citenamefont{Bermejo-Vega,
  Hangleiter, Schwarz, Raussendorf, and Eisert}}]{bermejo2017architectures}
\bibinfo{author}{\bibfnamefont{J.}~\bibnamefont{Bermejo-Vega}},
  \bibinfo{author}{\bibfnamefont{D.}~\bibnamefont{Hangleiter}},
  \bibinfo{author}{\bibfnamefont{M.}~\bibnamefont{Schwarz}},
  \bibinfo{author}{\bibfnamefont{R.}~\bibnamefont{Raussendorf}},
  \bibnamefont{and} \bibinfo{author}{\bibfnamefont{J.}~\bibnamefont{Eisert}},
  \bibinfo{journal}{arXiv:1703.00466}  (\bibinfo{year}{2017}).

\bibitem[{\citenamefont{Miller et~al.}(2017)\citenamefont{Miller, Sanders, and
  Miyake}}]{miller2017quantum}
\bibinfo{author}{\bibfnamefont{J.}~\bibnamefont{Miller}},
  \bibinfo{author}{\bibfnamefont{S.}~\bibnamefont{Sanders}}, \bibnamefont{and}
  \bibinfo{author}{\bibfnamefont{A.}~\bibnamefont{Miyake}},
  \bibinfo{journal}{arXiv:1703.11002}  (\bibinfo{year}{2017}).

\bibitem[{\citenamefont{Bremner et~al.}(2010)\citenamefont{Bremner, Jozsa, and
  Shepherd}}]{bremner2010classical}
\bibinfo{author}{\bibfnamefont{M.~J.} \bibnamefont{Bremner}},
  \bibinfo{author}{\bibfnamefont{R.}~\bibnamefont{Jozsa}}, \bibnamefont{and}
  \bibinfo{author}{\bibfnamefont{D.~J.} \bibnamefont{Shepherd}}, in
  \emph{\bibinfo{booktitle}{Proceedings of the Royal Society of London A:
  Mathematical, Physical and Engineering Sciences}} (\bibinfo{organization}{The
  Royal Society}, \bibinfo{year}{2010}), p. \bibinfo{pages}{0301}.

\bibitem[{\citenamefont{Bremner et~al.}(2016)\citenamefont{Bremner, Montanaro,
  and Shepherd}}]{bremner2016average}
\bibinfo{author}{\bibfnamefont{M.~J.} \bibnamefont{Bremner}},
  \bibinfo{author}{\bibfnamefont{A.}~\bibnamefont{Montanaro}},
  \bibnamefont{and} \bibinfo{author}{\bibfnamefont{D.~J.}
  \bibnamefont{Shepherd}}, \bibinfo{journal}{Physical Review Letters}
  \textbf{\bibinfo{volume}{117}}, \bibinfo{pages}{080501}
  (\bibinfo{year}{2016}).

\bibitem[{\citenamefont{Bremner et~al.}(2017)\citenamefont{Bremner, Montanaro,
  and Shepherd}}]{bremner2017achieving}
\bibinfo{author}{\bibfnamefont{M.~J.} \bibnamefont{Bremner}},
  \bibinfo{author}{\bibfnamefont{A.}~\bibnamefont{Montanaro}},
  \bibnamefont{and} \bibinfo{author}{\bibfnamefont{D.~J.}
  \bibnamefont{Shepherd}}, \bibinfo{journal}{Quantum}
  \textbf{\bibinfo{volume}{1}}, \bibinfo{pages}{8} (\bibinfo{year}{2017}).

\bibitem[{\citenamefont{Lloyd and Braunstein}(1999)}]{lloyd1999quantum}
\bibinfo{author}{\bibfnamefont{S.}~\bibnamefont{Lloyd}} \bibnamefont{and}
  \bibinfo{author}{\bibfnamefont{S.~L.} \bibnamefont{Braunstein}},
  \bibinfo{journal}{Physical Review Letters} \textbf{\bibinfo{volume}{82}},
  \bibinfo{pages}{1784} (\bibinfo{year}{1999}).

\bibitem[{\citenamefont{Braunstein and
  Van~Loock}(2005)}]{braunstein2005quantum}
\bibinfo{author}{\bibfnamefont{S.~L.} \bibnamefont{Braunstein}}
  \bibnamefont{and}
  \bibinfo{author}{\bibfnamefont{P.}~\bibnamefont{Van~Loock}},
  \bibinfo{journal}{Reviews of Modern Physics} \textbf{\bibinfo{volume}{77}},
  \bibinfo{pages}{513} (\bibinfo{year}{2005}).

\bibitem[{\citenamefont{Gu et~al.}(2009)\citenamefont{Gu, Weedbrook, Menicucci,
  Ralph, and van Loock}}]{gu2009quantum}
\bibinfo{author}{\bibfnamefont{M.}~\bibnamefont{Gu}},
  \bibinfo{author}{\bibfnamefont{C.}~\bibnamefont{Weedbrook}},
  \bibinfo{author}{\bibfnamefont{N.~C.} \bibnamefont{Menicucci}},
  \bibinfo{author}{\bibfnamefont{T.~C.} \bibnamefont{Ralph}}, \bibnamefont{and}
  \bibinfo{author}{\bibfnamefont{P.}~\bibnamefont{van Loock}},
  \bibinfo{journal}{Physical Review A} \textbf{\bibinfo{volume}{79}},
  \bibinfo{pages}{062318} (\bibinfo{year}{2009}).

\bibitem[{\citenamefont{Menicucci et~al.}(2006)\citenamefont{Menicucci, van
  Loock, Gu, Weedbrook, Ralph, and Nielsen}}]{menicucci2006universal}
\bibinfo{author}{\bibfnamefont{N.~C.} \bibnamefont{Menicucci}},
  \bibinfo{author}{\bibfnamefont{P.}~\bibnamefont{van Loock}},
  \bibinfo{author}{\bibfnamefont{M.}~\bibnamefont{Gu}},
  \bibinfo{author}{\bibfnamefont{C.}~\bibnamefont{Weedbrook}},
  \bibinfo{author}{\bibfnamefont{T.~C.} \bibnamefont{Ralph}}, \bibnamefont{and}
  \bibinfo{author}{\bibfnamefont{M.~A.} \bibnamefont{Nielsen}},
  \bibinfo{journal}{Physical Review Letters} \textbf{\bibinfo{volume}{97}},
  \bibinfo{pages}{110501} (\bibinfo{year}{2006}).

\bibitem[{\citenamefont{Douce et~al.}(2017)\citenamefont{Douce, Markham,
  Kashefi, Diamanti, Coudreau, Milman, van Loock, and
  Ferrini}}]{douce2017continuous}
\bibinfo{author}{\bibfnamefont{T.}~\bibnamefont{Douce}},
  \bibinfo{author}{\bibfnamefont{D.}~\bibnamefont{Markham}},
  \bibinfo{author}{\bibfnamefont{E.}~\bibnamefont{Kashefi}},
  \bibinfo{author}{\bibfnamefont{E.}~\bibnamefont{Diamanti}},
  \bibinfo{author}{\bibfnamefont{T.}~\bibnamefont{Coudreau}},
  \bibinfo{author}{\bibfnamefont{P.}~\bibnamefont{Milman}},
  \bibinfo{author}{\bibfnamefont{P.}~\bibnamefont{van Loock}},
  \bibnamefont{and} \bibinfo{author}{\bibfnamefont{G.}~\bibnamefont{Ferrini}},
  \bibinfo{journal}{Physical Review Letters} \textbf{\bibinfo{volume}{118}},
  \bibinfo{pages}{070503} (\bibinfo{year}{2017}).

\bibitem[{\citenamefont{Stroud}(1971)}]{stroud1971approximate}
\bibinfo{author}{\bibfnamefont{A.~H.} \bibnamefont{Stroud}},
  \bibinfo{journal}{Prentice-Hall}  (\bibinfo{year}{1971}).

\bibitem[{\citenamefont{Sloan and Wozniakowski}(1998)}]{sloan1998quasi}
\bibinfo{author}{\bibfnamefont{I.~H.} \bibnamefont{Sloan}} \bibnamefont{and}
  \bibinfo{author}{\bibfnamefont{H.}~\bibnamefont{Wozniakowski}},
  \bibinfo{journal}{Journal of Complexity} \textbf{\bibinfo{volume}{14}},
  \bibinfo{pages}{1} (\bibinfo{year}{1998}).

\bibitem[{\citenamefont{Novak and Wozniakowski}(2009)}]{novak2009approximation}
\bibinfo{author}{\bibfnamefont{E.}~\bibnamefont{Novak}} \bibnamefont{and}
  \bibinfo{author}{\bibfnamefont{H.}~\bibnamefont{Wozniakowski}},
  \bibinfo{journal}{Journal of Complexity} \textbf{\bibinfo{volume}{25}},
  \bibinfo{pages}{398} (\bibinfo{year}{2009}).

\bibitem[{\citenamefont{Novak and Wozniakowski}(2008)}]{novak2008tractability1}
\bibinfo{author}{\bibfnamefont{E.}~\bibnamefont{Novak}} \bibnamefont{and}
  \bibinfo{author}{\bibfnamefont{H.}~\bibnamefont{Wozniakowski}},
  \bibinfo{journal}{EMS Tracts in Mathematics} \textbf{\bibinfo{volume}{6}}
  (\bibinfo{year}{2008}).

\bibitem[{\citenamefont{Novak and Wozniakowski}(2010)}]{novak2010tractability2}
\bibinfo{author}{\bibfnamefont{E.}~\bibnamefont{Novak}} \bibnamefont{and}
  \bibinfo{author}{\bibfnamefont{H.}~\bibnamefont{Wozniakowski}},
  \bibinfo{journal}{EMS Tracts in Mathematics} \textbf{\bibinfo{volume}{12}}
  (\bibinfo{year}{2010}).

\bibitem[{\citenamefont{Hinrichs et~al.}(2014)\citenamefont{Hinrichs, Novak,
  Ullrich, and Wo{\'z}niakowski}}]{curse}
\bibinfo{author}{\bibfnamefont{A.}~\bibnamefont{Hinrichs}},
  \bibinfo{author}{\bibfnamefont{E.}~\bibnamefont{Novak}},
  \bibinfo{author}{\bibfnamefont{M.}~\bibnamefont{Ullrich}}, \bibnamefont{and}
  \bibinfo{author}{\bibfnamefont{H.}~\bibnamefont{Wo{\'z}niakowski}},
  \bibinfo{journal}{Mathematics of Computation} \textbf{\bibinfo{volume}{83}},
  \bibinfo{pages}{2853} (\bibinfo{year}{2014}).

\bibitem[{\citenamefont{Rohde et~al.}(2016)\citenamefont{Rohde, Berry, Motes,
  and Dowling}}]{rohde2016quantum}
\bibinfo{author}{\bibfnamefont{P.~P.} \bibnamefont{Rohde}},
  \bibinfo{author}{\bibfnamefont{D.~W.} \bibnamefont{Berry}},
  \bibinfo{author}{\bibfnamefont{K.~R.} \bibnamefont{Motes}}, \bibnamefont{and}
  \bibinfo{author}{\bibfnamefont{J.~P.} \bibnamefont{Dowling}},
  \bibinfo{journal}{arXiv:1607.04960}  (\bibinfo{year}{2016}).

\bibitem[{\citenamefont{Bartlett et~al.}(2002)\citenamefont{Bartlett, Sanders,
  Braunstein, and Nemoto}}]{bartlett2002efficient}
\bibinfo{author}{\bibfnamefont{S.~D.} \bibnamefont{Bartlett}},
  \bibinfo{author}{\bibfnamefont{B.~C.} \bibnamefont{Sanders}},
  \bibinfo{author}{\bibfnamefont{S.~L.} \bibnamefont{Braunstein}},
  \bibnamefont{and} \bibinfo{author}{\bibfnamefont{K.}~\bibnamefont{Nemoto}},
  \bibinfo{journal}{Physical Review Letters} \textbf{\bibinfo{volume}{88}},
  \bibinfo{pages}{097904} (\bibinfo{year}{2002}).

\bibitem[{\citenamefont{Toda}(1991)}]{toda1991pp}
\bibinfo{author}{\bibfnamefont{S.}~\bibnamefont{Toda}}, \bibinfo{journal}{SIAM
  Journal on Computing} \textbf{\bibinfo{volume}{20}}, \bibinfo{pages}{865}
  (\bibinfo{year}{1991}).

\bibitem[{\citenamefont{Stockmeyer}(1985)}]{stockmeyer1985approximation}
\bibinfo{author}{\bibfnamefont{L.}~\bibnamefont{Stockmeyer}},
  \bibinfo{journal}{SIAM Journal on Computing} \textbf{\bibinfo{volume}{14}},
  \bibinfo{pages}{849} (\bibinfo{year}{1985}).

\end{thebibliography}
\pagebreak

\onecolumngrid

\section{Supplemental Material}

\subsection{CV-IQP integrals}

We perform a detailed derivation of the expressions for CV-IQP integrals. In the ideal case, input states are infinitely momentum-squeezed vacuum states given by
\beq
\ket{0}_p=\frac{1}{\sqrt{2\pi}}\int_{\mathbb{R}}\ket{q}dq.
\eeq
If the measurements are homodyne with infinite precision, the probability amplitude of obtaining an outcome $s=(s_1,s_2,\ldots,s_n)$ after the action of a diagonal circuit $C_f$ is given by
\begin{align}\label{A_f}
A_f(s)&=\bra{s}_pC_f\ket{00\ldots0}_p\nonumber\\
&=\frac{1}{(2\pi)^n}\int_{\mathbb{R}^n}e^{i f(q)}e^{-i s\cdot q'}\langle q'|q\rangle dq'^ndq^n\nonumber\\
&=\frac{1}{(2\pi)^n}\int_{\mathbb{R}^n}e^{i f(q)}e^{-i s\cdot q}dq^n.
\end{align}
When the inputs are finitely squeezed states of the form
\beq
\ket{\sigma}_k=\frac{1}{\pi^{1/4}\sqrt{\sigma}}\int_{\mathbb{R}}e^{- q_k^2/(2\sigma^2)}\ket{q}dq_k
\eeq
with variance $\sigma^2$, the action of the circuit $C_f$ on the inputs produces the state
\beq
\ket{f}=\frac{1}{\sqrt{\mathcal{N}^n}}\int_{\mathbb{R}}e^{if(q)}e^{- q^2/(2\sigma^2)}\ket{q}dq
\eeq
where we have defined $\mathcal{N}=1/\sqrt{\pi}\sigma$. If we then apply a measurement with limited precision $\Delta_p$ given by the projectors
\beq
\eta_{s_k}=\int_{s_k-\Delta_p}^{s_k+\Delta_p}\ket{p_k}\bra{p_k},
\eeq 
the probability of obtaining an outcome $s$ is given by
\begin{align*}
\tilde{P}_f(s)&=\bra{f}\bigotimes_k\eta_{s_k}\ket{f}\\
&=\frac{1}{\mathcal{N}^n}\int_{\mathbb{R}^n}e^{i f(q)-if(q')}e^{- (q^2-q'^2)/(2\sigma^2)}\left[\prod_k \int_{s_k-\Delta_p}^{s_k+\Delta_p}\langle q'_k|p_k\rangle\langle p_k|q_k\rangle dp_k\right]dq^ndq'^n\\
&=\frac{1}{(2\pi \mathcal{N})^n}\int_{\mathbb{R}^n}e^{i f(q)-if(q')}e^{- (q^2-q'^2)/(2\sigma^2)}\left[\prod_k \int_{s_k-\Delta_p}^{s_k+\Delta_p}e^{ip_k(q'_k-q_k)} dp_k\right]dq^ndq'^n\\
&=\frac{(\Delta_p)^n}{(\pi \mathcal{N})^n}\int_{\mathbb{R}^n}e^{i f(q)-if(q')}e^{- (q^2-q'^2)/(2\sigma^2)}e^{-is\cdot(q-q')}\prod_k \sinc[\Delta_p(q_k-q'_k)]dq^ndq'^n.
\end{align*}
In practice, it is straightforward to obtain very high precision in homodyne measurements while it is challenging to obtain large values of squeezing. Thus, in the regime where $\sigma\gg 1/\Delta_p$, the sinc functions are approximately equal to 1 in the region of non-negligible values of the integrand and we can write $\tilde{P}_f(s)=|\tilde{A}_f(s)|^2$ where
\beq
\tilde{A}_f(s)=\left(\frac{\Delta_p}{\pi^{3/2}\sigma}\right)^{n/2}\int_{\mathbb{R}^n}e^{i f(q)-i s\cdot q}e^{-q^2/(2\sigma^2)}dq^n.
\eeq

\subsection{CV-IQP integrals as weighted sums of \#P-hard problems}
Recall the expression of Eq. \eqref{A_f} for the probability amplitude of obtaining an outcome $s=(s_1,s_2,\ldots,s_n)$ in a CV-IQP circuit. Let $f_s(q):=f(q)-s\cdot q$. For large values of $q$, the integrand is highly oscillatory and the integral averages to zero. This means that we can make the approximation
\beq
A_f(s)= \frac{1}{(2\pi)^n}\int_{D_L}e^{i f_s(q)}dq^n+\epsilon_a
\eeq
for some appropriately chosen constant $L$, where $D_L=[-L,L]^n$ is a hypercube of length $2L$ centered at the origin and $\epsilon_a$ is an arbitrarily small error due to the finite integration region. We can further approximate this integral by a Riemann sum over a rectangular lattice
\begin{align*}
A_f(s)&= \frac{1}{(2\pi)^n}\sum_{q_1,\ldots ,q_n}e^{i f_s(q_1,\ldots ,q_n)}\Delta q^n+\epsilon_a+\epsilon_b
\end{align*}
where $q_i\in\{-L,-L+\Delta q,\ldots,L-2\Delta q,L-\Delta q\}$ and $\epsilon_b$ is the arbitrarily small error in the approximation. Our goal is to further approximate the integrand function by a series of step functions. Let $b\in\{0,1,\ldots,N-1\}$ with $N=2^l$ for some integer $l>0$ and define the angles
\beq
\phi_b=\frac{ 2\pi b }{N}
\eeq
as well as the indicator functions
\beq
\Theta_b^{f_s}(q_1,\ldots,q_n)=\begin{cases}
1& \textrm{ if } f_s(q)\in [\phi_b,\phi_{b+1})\mod 2\pi\\
0 & \textrm{ otherwise}
\end{cases}
\eeq
so that we can approximate the complex exponential of the polynomial as
\beq
e^{if_s(q)}\simeq \sum_{b=0}^{N-1}e^{i \phi_b}\Theta_b^{f_s}(q).
\eeq
There are $2L/\Delta q$ possible values of each $q_i$, so if we set $2L/\Delta q=2^k$ for some integer $k$, we can associate each vector $(q_1,\ldots ,q_n)$ with an $m$-bit string $x\in\{0,1\}^m$, $m=kn$. Consequently, we can view each $\Theta_b^{f_s}$ as a Boolean function $\Theta_b^{f_s}:\{0,1\}^m\rightarrow \{0,1\}$. In this case, we write the approximation of the integral as
\begin{align}
A_f(s)&=\frac{\Delta q^n}{(2\pi)^n}\sum_{b=0}^{N-1}e^{i \phi_b} \sum_{x\in\{0,1\}^m}\Theta_b^{f_s}(x)+\epsilon_a+\epsilon_b+\epsilon_c\nonumber\\
&=\frac{\Delta q^n}{(2\pi)^n}\sum_{b=0}^{N-1}\sum_{x\in\{0,1\}^m}e^{i \phi_b} \Theta_b^{f_s}(x)+\epsilon,
\end{align}
where $\epsilon=\epsilon_a+\epsilon_b+\epsilon_c$ and $\epsilon_c$ is the error arising from the step-function approximation of $e^{if_s(q)}$.

\subsection{Exponential running time for deterministic numerical integration}
We follow the results of Ref. \cite{curse} and consider, without loss of generality, a numerical integration algorithm $A$ that uses a fixed set of $n$-dimensional sampling points $\theta_1,\ldots,\theta_N$ and an arbitrary mapping $\Lambda$ to approximate the integral of a function $\phi(q)$ as
\beq
A(\phi)=\Lambda[\phi(\theta_1),\ldots,\phi(\theta_N)].
\eeq
The error in the approximation is defined as
\beq
\varepsilon=\left|\int\phi(q)dq^n-A(\phi)\right|.
\eeq
To bound the number of function calls $N$, we define the fooling function
\beq
\phi(q)=\min\left\{1,\frac{1}{\sqrt{n}\delta}\text{dist}(q,\Gamma_{\delta})\right\}
\eeq
where 
\beq
\Gamma_{\delta}=\bigcup_{i=1}^NB_{\delta}(\theta_i),
\eeq
$\text{dist}(\cdot\,,\cdot)$ is the Euclidean distance, and $B_{\delta}(\theta_i)$ is a ball with radius $\delta\sqrt{n}$ centered at the point $\theta_i$. By construction, the fooling function satisfies $\phi(\theta_i)=0$ for all $i=1,\ldots,N$ and therefore the algorithm must output the same approximation for $\phi(q)$ and $-\phi(q)$. This allows us to bound the additive error $\epsilon$ in the approximation of the algorithm $A$ in terms of the value of the integral of $\phi(q)$ as \cite{curse}
\begin{align}
\varepsilon&\geq\frac{1}{2}\left|\int_{\mathbb{R}^n}\phi(q)dq^n-A(\phi)\right|+\left|\int_{\mathbb{R}^n}-\phi(q)dq^n-A(\phi)\right|\\
&\geq \frac{1}{2}\left|\int_{\mathbb{R}^n}\phi(q)dq^n\right|+\frac{1}{2}\left|\int_{\mathbb{R}^n}-\phi(q)dq^n\right|\nonumber\\
&=\int_{\mathbb{R}^n}\phi(q)dq^n.
\end{align}
In fact, this bound also holds if we multiply $\phi(q)$ by any strictly positive function $g(q)> 0$ since the algorithm $A$ also gives the same answer for $\phi(q)g(q)$ and $-\phi(q)g(q)$, so we can write
\beq\label{epsilon_int}
\varepsilon\geq \int_{\mathbb{R}^n}\phi(q)g(q)dq^n.
\eeq
Now recall the expression for the CV-IQP integral in the presence of finite squeezing, where we are omitting known normalization factors
\beq
A_f(s)= \int_{\mathbb{R}^n}e^{i f(q)-is\cdot q}e^{- q^2/(2\sigma^2)}dq^n
\eeq
with the real part of the integral given by
\beq
\text{Re}[A_f(s)]= \int_{\mathbb{R}^n}\cos[f(q)-is\cdot q]e^{- q^2/(2\sigma^2)} dq^n.
\eeq
We now fix $f(q)$ to satisfy the relation
\beq
f(q)-s\cdot q=\cos^{-1}[\phi(q)]:=\tilde{f}_s(q).
\eeq 
Note that $0\leq \phi(q)\leq 1$ and therefore the inverse cosine is well defined. Moreover, since we have made no restrictions about diagonal circuits implementing $e^{i f(q)}$, we take $\tilde{f}_s(q)$ to be an arbitrarily good polynomial approximation of $\cos^{-1}[\phi(q)]$. We then have 
\begin{align}
\text{Re}[A_f(s)]=&\int_{\mathbb{R}^n}\cos[\tilde{f}_s(q)]e^{- q^2/(2\sigma^2)}\nonumber\\
&= \int_{\mathbb{R}^n}\phi(q)e^{- q^2/(2\sigma^2)}dq^n+\epsilon_a\nonumber\\
&=\int_{\mathbb{R}^n}\phi(q)g(q)dq^n+\epsilon_a
\end{align}
where $\epsilon_a$ is the arbitrarily small error arising from the polynomial approximation of $\cos^{-1}[\phi(q)]$ and we have implicitly defined
\beq
g(q):=e^{- q^2/(2\sigma^2)}.
\eeq
The function $g(q)$ is exponentially decreasing for large $q$ which allows us to write
\begin{align}\label{real to fooling}
\text{Re}[A_f(s)]&=\int_{D_L}\phi(q)g(q)dq^n+\epsilon_a+\epsilon_b
\end{align} 
where, as before, $D_L=[-L,L]^n$ is a hypercube of length $2L$ centered at the origin and $\epsilon_b$ is an exponentially small error arising from the approximation due to the finite integration region. From Eq. \eqref{real to fooling} we conclude that the algorithm $A$ must also, up to negligible errors $\epsilon_a$ and $\epsilon_b$, incur an error $\epsilon$ in evaluating the real part of the CV-IQP integral and therefore an error at least $\epsilon$ in evaluating the full complex value of the integral.\\

We now proceed to give a lower bound on the number of sampling points that are needed to achieve a fixed error $\epsilon$ in the numerical integration of the fooling function $\phi(q)g(q)$ and therefore also on the CV-IQP integral. We have that
\begin{align}
\int_{D_L}\phi(q)g(q)dq^n&\geq \int_{D_L/P_\delta}\phi(q)g(q)dq^n\nonumber\\
&=\int_{D_L/P_\delta}g(q)dq^n,
\end{align}
where we have used the fact that $\phi(q)g(q)\geq 0$ for $q\in P_\delta$ and $\phi(q)=1$ for $q\notin P_\delta$. Since $g(q)\leq 1$ for all $q$, it holds that
\begin{align}
\int_{D_L/P_\delta}g(q)dq^n&\geq\int_{D_L}g(q)dq^n-\text{Vol}(P_\delta)\nonumber\\
&> \int_{D_L}g(q)dq^n-N(\delta\sqrt{2\pi e})^n\nonumber\\
&=\int_{\mathbb{R}^n}g(q)dq^n-N(\delta\sqrt{2\pi e})^n+\epsilon_b\nonumber\\
&=(\sqrt{2\pi}\sigma)^n-N(\delta\sqrt{2\pi e})^n+\epsilon_b,
\end{align}
where we used the bound $\text{Vol}(P_\delta)<(\delta\sqrt{2\pi e})^n$ for the volume of an n-dimensional ball with radius $\delta$. Combining these results we obtain
\begin{align}
\epsilon\geq(\sqrt{2\pi}\sigma)^n-N(\delta\sqrt{2\pi e})^n+2\epsilon_b.
\end{align}
Combining this with Eq. \eqref{epsilon_int} and recalling that $\epsilon_b$ is the error arising from the finite integration region $D_L$ gives
\beq
N\geq\left[\frac{\sigma}{\delta\sqrt{e}}\right]^n-(\epsilon-2\epsilon_b).
\eeq
Therefore, as long as we choose $\delta<\sigma/\sqrt{e}$ for the fooling function, it will hold that $N=2^{O(n)}$. We conclude that there exists a class of worst-case fooling functions $\phi(q)$ and corresponding CV-IQP circuits such that to achieve a constant approximation error in evaluating the corresponding CV-IQP integral, any classical numerical algorithm $A$ requires an exponential number of functions calls and therefore an exponentially large running time.

\subsection{Hardness of Sampling}

We begin with a proof of Lemma 1 in the main manuscript, which we reproduce here for clarity.\\

\begin{lem}\label{lemma} Let $C_f$ be a CV-IQP circuit acting on $n$ qumodes, where $C_f$ is chosen from some appropriate family of circuits. Let $C_{f,r}$ be the circuit obtained by adding diagonal gates $U_r=\prod_{k=1}^ne^{-i q_k r_k}$ to $C_f$, with $r=(r_1,r_2,\ldots,r_n)$ and $r_k\in\{-L,-L+2\Delta_p,\ldots,L-2\Delta_p,L\}$. Assume that there exists a polynomial-time classical algorithm $A$ such that for any CV-IQP circuit $C'_f$, the algorithm $A$ can approximate the probability distribution of $C'_f$ up to additive error $\epsilon$. Then for any $\delta>0$, there exists an $\text{FBPP}^{\text{NP}}$ algorithm that given access to $A$ approximates $|A_{f,r}(0)|^2$ for a circuit $C_{f,r}$ up to additive error
\beq
O\left(\frac{(1+o(1))\epsilon}{\delta \ell^n}+|A_{f,r}(0)|^2/\text{poly(n)}\right)
\eeq
with probability at least $1-\delta$.\\
\end{lem}

\textit{Proof:} We follow closely the proof of Ref. \cite{bremner2016average} adapted to CV-IQP circuits. Define $p_{rs}=\Pr(C_{f,r}\text{ outputs } s)$. Similarly, define $q_{rs}=\Pr(A\text{ outputs } s)$ on input $C_{f,r}$, denoting by $q_{0s}$ the probability for $r=(0,0,\ldots,0)$. From Stockmeyer's counting algorithm \cite{stockmeyer1985approximation}, there exists an $\text{FBPP}^{\text{NP}}$ algorithm with access to $A$ that produces an estimate $\tilde{q}_{s}$ such that
\beq
|\tilde{q}_{s}-q_{0s}|\leq q_{0s}/\text{poly(n)}.
\eeq
Then it holds that 
\begin{align*}
|\tilde{q}_{s}-p_{0s}|&\leq |\tilde{q}_{s}-q_{0s}|+|q_{0s}-p_{0s}|\\
&\leq q_{0s}/\text{poly(n)}+|q_{0s}-p_{0s}|\\
&\leq(p_{0s}+|q_{0s}- p_{0s}|)/\text{poly(n)}+|q_{0s}-p_{0s}|\\
&=p_{0s}/\text{poly(n)}+|q_{0s}-p_{0s}|(1+1/\text{poly(n)}).
\end{align*}
From Markov's inequality we have
\begin{align*}
\Pr_s\left(|q_{0s}-p_{0s}|\geq \frac{\epsilon}{\delta \ell^n}\right)&\leq \frac{\mathbb{E}(|q_{0s}-p_{0s}|)\delta \ell^n}{\epsilon}\\
&=\frac{1}{\ell^n}\frac{\sum_r(|q_{0s}-p_{0s}|)\delta \ell^n}{\epsilon}\\
&=\delta,
\end{align*}
where $r$ is chosen uniformly at random and we have used the fact that $\sum_r(|q_{0s}-p_{0s}|)=\epsilon$. Therefore with probability at least $1-\delta$,
\beq
|\tilde{q}_{s}-p_{0s}|\leq p_{0s}/\text{poly(n)}+\frac{\epsilon(1+1/\text{poly(n)})}{\delta \ell^n}.
\eeq
Finally, $ p_{0s}=\Pr(C_{f,0}\text{ outputs } s)=\Pr(C_{f,s}\text{ outputs } 0)=|A_{f,s}(0)|^2$ and we conclude that, with probability at least $1-\delta$,
\beq
|\tilde{q}_{s}-|A_{f,s}(0)|^2|\leq \frac{\epsilon(1+1/\text{poly(n)})}{\delta \ell^n}+|A_{f,s}(0)|^2/\text{poly(n)}
\eeq
as desired. \qed\\\\

We now prove an anti-concentration result for CV-IQP integrals with finite squeezing and precision. The same statement holds as well for integrals in the ideal case. From the Payley-Zigmund inequality, it holds that 
\beq\label{Payley}
\Pr(|A_f(0)|^2\geq \alpha\mathbb{E}[|A_f(0)|^2])\geq (1-\alpha)^2\frac{\mathbb{E}[|A_f(0)|^2]^2}{\mathbb{E}[|A_f(0)|^4]},
\eeq
where the expectation is taken over all circuits $C_{f,s}$ in the corresponding family. This family can be, for instance, defined as the class of circuits that leads to CV-IQP integrals that approximate weighted sums of \#P-hard problems. Following Ref. \cite{bremner2016average}, we have that
\begin{align*}
\mathbb{E}[|A_f(0)|^2]&=\mathbb{E}_{C_{f,s}}[|A_{f}(0)|^2]\\
&=\mathbb{E}_{s,C_f}[\sum_s|A_f(0)|^2]\\
&=\frac{1}{\ell^n}\mathbb{E}_{C_f}[\sum_s|A_f(0)|^2]=\frac{1}{\ell^n}+o(1)
\end{align*}
where the correction $o(1)$ comes from the fact that the values of $s$ are restricted to the finite integration region. This gives an expression for the numerator in the right-hand side of Eq. \eqref{Payley}. For our purposes, it suffices to upper bound the denominator. We have  
\begin{align*}
\mathbb{E}[|A_f(0)|^4]&=\mathbb{E}_{C_f}[|A_f(0)|^4]\\
&\leq \max_{C_f} |A_f(0)|^4
\end{align*}
so we just need to upper bound $|A_f(0)|$ for all circuits. From Eq. \eqref{A_f}, the integrand is upper bounded in absolute value by 1 and therefore
\begin{align}
|A_f(0)|&=|\frac{1}{(2\pi)^n}\int_{D_L}e^{i f(q)}e^{-is\cdot q}dq^n|\nonumber\\
&\leq \frac{1}{(2\pi)^n}\text{Vol}(D_L)\nonumber\\
&=\left(\frac{L}{\pi}\right)^n
\end{align}
where we used $\text{Vol}(D_L)=(2L)^n$. We then have that
\beq
\frac{\mathbb{E}[|A_f(0)|^2]^2}{\mathbb{E}[|A_f(0)|^4]}\geq\left(\frac{L}{\pi \sqrt{\ell}}\right)^{4n}=\left(\frac{\sqrt{\Delta_p L}}{\pi}\right)^{4n}
\eeq
where we have replaced $\ell=L/\Delta_p$. The expression on the right hand side can be made equal to a constant by fixing $L$, a free parameter of our choosing, appropriately. In that case, by setting $L$ such that
\beq\label{Lcondition}
L=\frac{\pi}{(\Delta_p)^2}
\eeq
for a given $\Delta_p$, we have
\beq
\frac{\mathbb{E}[|A_f(0)|^2]^2}{\mathbb{E}[|A_f(0)|^4]}\geq 1
\eeq
and therefore
\begin{align}\label{anti-con}
&\Pr(|A_f(0)|^2\geq \alpha\mathbb{E}[|A_f(0)|^2])\nonumber\\
=&\Pr(|A_f(0)|^2\geq \alpha \ell^{-n})\geq (1-\alpha)^2.
\end{align}

The condition of Eq. \eqref{Lcondition} together with a large enough value of $L$ for a good approximation to the integral can both be met simultaneously provided that $\Delta_p$ is small enough. Note that even if equality does not hold exactly in Eq. \eqref{Lcondition}, as long as
\beq
\left(\frac{\sqrt{\Delta_p L}}{\pi}\right)=1+O(1/n)
\eeq
we have that
\beq
\left(\frac{\sqrt{\Delta_p L}}{\pi}\right)^{4n}=\left(1+O(1/n)\right)^{4n}=O(1).
\eeq
Thus, from now on we assume that $L$ satisfies Eq. \eqref{Lcondition} and statements hold for integrals over the corresponding hypercube. Integrals with a different value of $L$ will themselves be exponentially good approximations to this one, so the following results would apply for such integrals as well. By setting $\alpha=\frac{1}{2}$, from Eq. \eqref{anti-con}, we conclude that
\beq
\Pr(|A_f(0)|^2\geq \frac12\ell^{-n})\geq \frac{1}{4}
\eeq 
as desired for the anti-concentration result.

\end{document}